\documentclass[journal, twocolumn]{IEEEtran}

\usepackage{amsmath, amssymb, amsthm}
\usepackage{mathtools}
\usepackage{physics}
\usepackage{pgfplots}
\pgfplotsset{compat=1.18}
\usepackage{hyperref}
\usepackage{paralist}

\usepackage{tikz}
\usetikzlibrary{arrows.meta,calc}

\usepackage{etoolbox}
\usepackage{enumitem}

\usepackage{pgf}
\usepackage{subcaption}

\usepackage{cite}

\usepackage{bm}

\newtheorem{theorem}{Theorem}[section]
\newtheorem{lemma}[theorem]{Lemma}
\newtheorem{proposition}[theorem]{Proposition}
\newtheorem{corollary}[theorem]{Corollary}

\theoremstyle{definition}
\newtheorem{definition}{Definition}[section]
\theoremstyle{remark}
\newtheorem*{remark}{Remark}
\theoremstyle{definition}

\AtEndEnvironment{example}{\hfill$\square$}

\hypersetup{hidelinks}

\begin{document}

\title{Asymptotic Behavior of an Unforced Duhem-Type Hysteretic Oscillator}

\author{Mihails Milehins and Dan B. Marghitu%
\thanks{Both authors are with the Department of Mechanical Engineering, Auburn University, Auburn, AL 36849 (corresponding author: Mihails Milehins; e-mail: mzm0390@auburn.edu).}}

\maketitle

\begin{abstract}
The Duhem hysteresis models encompass a broad class of phenomenological models used to capture viscoelastoplastic behaviors in a wide variety of physical systems. While these models are utilized across many fields of science and technology, their inherent non-smoothness introduces substantial analytical complexities. This article presents a qualitative analysis of the fundamental analytical properties of an unforced 1-DOF mechanical oscillator incorporating a Duhem-type viscoelastoplastic hysteretic element. The proposed mathematical formulation is designed to accommodate a wide range of system parameters and initial conditions without imposing excessive regularity assumptions. Direct analysis is employed to systematically establish the global existence of solutions, uniqueness of solutions, and convergence of each solution to an equilibrium point. The latter result is a significant generalization of the existing results in the research literature. All three results have been formally verified in the Lean 4 proof assistant. Ultimately, this work demonstrates that the asymptotic behavior of the generalized Duhem-type oscillator is consistent with phenomenological observations, providing a robust theoretical foundation for its application in nonlinear dynamics.
\end{abstract}


\section{Introduction}\label{sec:introduction}

The Duhem models constitute a class of differential models of hysteresis that are suitable for the description of a variety of physical phenomena, including dry friction, elastoplastic materials, and magnetization. The models are usually attributed to Pierre Duhem [Duhem (1896-1902), as cited in \citenum{ikhouane_survey_2018}]. However, the interest in these models has surged only in the second half of the twentieth century \cite{everett_general_1954, chua_mathematical_1971, chua_generalized_1972, pokrovskii_theory_1973, visintin_continuity_1983, coleman_constitutive_1986, coleman_class_1987, krasnoselskii_systems_1989, visintin_differential_1994, oh_modeling_2003, oh_analysis_2003, oh_semilinear_2005, padthe_counterclockwise_2005, jayawardhana_sufficient_2009, jayawardhana_dissipativity_2011, jayawardhana_stability_2012, ouyang_stability_2012, naser_consistency_2013, naser_characterization_2013, vasquez-beltran_modeling_2023}. Specializations of the models include, for example, the Dahl friction model \cite{dahl_solid_1968}, the LuGre friction model \cite{canudas_de_wit_new_1995}, and the Bouc-Wen model \cite{bouc_forced_1968, bouc_modemathematique_1971, wen_method_1976}. Refs. \cite{visintin_differential_1994} and \cite{ikhouane_survey_2018, ikhouane_erratum_2018} provide further general information about the Duhem models. Refs. \cite{krasnoselskii_systems_1989, mayergoyz_mathematical_1991, visintin_differential_1994, brokate_hysteresis_1996, mielke_rate-independent_2015} provide further information about differential models of hysteresis.

The following form of a Duhem model is essentially equivalent\footnote{Apart from syntactic discrepancies, a different set of regularity conditions for $f$, $g$, and $h$ is employed in Ref. \cite{oh_analysis_2003}.} to the form employed in Ref. \cite{oh_analysis_2003}:\footnote{See Appendix \ref{sec:NCF} for notation and conventions.}
\begin{equation}\label{eq:Duhem_OB}
\begin{cases}
\dot{x} = u \\
\dot{z} = f(x, z) g(u) \\
y = - h(x, z) \\
\begin{matrix}
x(0) = x_0 & z(0) = z_0 
\end{matrix}
\end{cases}
\end{equation}
where $x \in \mathbb{R}$ and $z \in \mathbb{R}^n$ are state variables, $u \in \mathbb{R}$ is an input variable, $y \in \mathbb{R}$  is an output variable, $x_0 \in \mathbb{R}$, $z_0 \in \mathbb{R}^n$ are parameters, $f : \mathbb{R} \times \mathbb{R}^n \longrightarrow \mathbb{R}^{n \times p}$, $g : \mathbb{R} \longrightarrow \mathbb{R}^p$, $h : \mathbb{R} \times \mathbb{R}^n \longrightarrow \mathbb{R}$ are continuous functions, and $n, p \in \mathbb{Z}_{\geq 1}$.

In this article, the following specialization of the model given by Eq. \eqref{eq:Duhem_OB} will be considered:
\begin{equation}\label{eq:Duhem_SOB}
\begin{cases}
\dot{x} = u\\
\dot{z} = u + f_1(z) g_1(u) + f_2(z) g_2(u) \\
y = - h_1(x) - h_2(z)\\
\begin{matrix}
x(0) = x_0 & z(0) = z_0
\end{matrix}
\end{cases}
\end{equation}
Here, $x, z \in \mathbb{R}$ are state variables, $u \in \mathbb{R}$ is an input variable, $y \in \mathbb{R}$ is an output variable, $f_1, f_2, g_1, g_2 : \mathbb{R} \longrightarrow \mathbb{R}$ are locally Lipschitz continuous functions, and $h_1, h_2 : \mathbb{R} \longrightarrow \mathbb{R}$ are strictly increasing locally Lipschitz homeomorphisms such that
\begin{equation}
f_1(0) = f_2(0) = g_1(0) = g_2(0) = h_1(0) = h_2(0) = 0
\end{equation}
\begin{equation}
\forall z \in \mathbb{R}_{< 0}. \; 0 < f_1(z) \wedge f_2(z) < 0
\end{equation}
\begin{equation}
\forall z \in \mathbb{R}_{> 0}. \; f_1(z) < 0 \wedge 0 < f_2(z)
\end{equation}
\begin{equation}
\forall v \in \mathbb{R}_{> 0}. \; 0 < g_1(v) \wedge g_2(v) = 0
\end{equation}
\begin{equation}
\forall v \in \mathbb{R}_{< 0}. \; g_1(v) = 0 \wedge g_2(v) < 0
\end{equation}

The following model represents an unforced mechanical oscillator (e.g., see Ref. \cite{mickens_truly_2010}) with a viscoelastoplastic Duhem-type element and (optionally) additional viscous damping:\footnote{For simplicity, the parameter that represents the inertia of the system is ignored. The inertial properties of the system can be taken into account by rescaling or making appropriate adjustments to $h_1$, $h_2$, and $c$.} 
\begin{equation}\label{eq:main}
\begin{cases}
\dot{x} = v\\
\dot{z} = v + f_1(z) g_1(v) + f_2(z) g_2(v) \\
\dot{v} = - h_1(x) - h_2(z) - c(x, z, v)\\
\begin{matrix}
x(0) = x_0 & z(0) = z_0 & v(0) = v_0
\end{matrix}
\end{cases}
\end{equation}
Here, $x, z, v \in \mathbb{R}$ are state variables, and $x_0, z_0, v_0 \in \mathbb{R}$ are parameters, $c : \mathbb{R}^3 \longrightarrow \mathbb{R}$ is a locally Lipschitz continuous function such that
\begin{equation}
\forall x, z \in \mathbb{R}. \; c(x, z, 0) = 0
\end{equation}
\begin{equation}
\forall x, z, v \in \mathbb{R}. \; 0 \leq v c(x, z, v)
\end{equation}

The article addresses the problem of the asymptotic behavior of solutions of the system given by Eq. \eqref{eq:main}. More specifically, it is shown that every trajectory of the system converges to an equilibrium point. This result is a generalization of a result about the asymptotic behavior of the Bouc-Wen model of class I (except a single boundary case) that was previously presented in Refs. \cite{ikhouane_dynamic_2007, ikhouane_systems_2007}. It should be noted that other related problems were considered in \cite{yakubovich_method_1965, barabanov_absolute_1979, anderssen_global_1998, lacy_hysteretic_2000, logemann_systems_2003, logemann_class_2008, jayawardhana_stability_2012, ouyang_stability_2012, ouyang_absolute_2014}. 

The remainder of the article is organized as follows:
\begin{compactitem}
\item Section \ref{sec:formalization} describes the formalization of the results in the Lean 4 proof assistant.
\item Section \ref{sec:convergence} provides proofs of global existence, uniqueness and boundedness of solutions of Eq. \eqref{eq:main}.
\item Section \ref{sec:SC} provides a proof of convergence of each solution of Eq. \eqref{eq:main} to an equilibrium point. 
\item Section \ref{sec:BW} presents an application of the results developed in Section \ref{sec:SC} to a Bouc-Wen oscillator. 
\item Section \ref{sec:conclusions} provides conclusions and recommendations.
\item Appendix \ref{sec:NCF} describes the mathematical conventions and provides the proofs of secondary results.
\end{compactitem}

\section{Formal Verification}\label{sec:formalization}

The propositions, theorems, and corollary of Sections \ref{sec:convergence} and \ref{sec:SC}, together with the application of Theorem \ref{thm:main} to the Bouc-Wen model of class I with $\gamma \neq \beta$ in Section \ref{sec:BW}, have been formalized in the Lean 4 proof assistant \cite{de_moura_lean_2015, de_moura_lean_2021}, with its mathematical library Mathlib \cite{the_mathlib_community_lean_2020}.\footnote{See \url{https://gitlab.com/user9716869/LEAN-ADHO}.} The formalization follows the proofs of this article: the lemmas of these sections and Lemmas \ref{thm:con_sep} and \ref{thm:fgx} are formalized as steps of the proofs of the other results. The results that the proofs cite from the literature are not assumed: the special cases that the article needs, including the local existence, uniqueness, and extendability of solutions, Lyapunov's stability theorem, the properties of positive limit sets, and Danskin's theorem, are proved from the results of Mathlib. Every proof in the formalization is checked by the kernel of Lean, and the formalization uses no axioms other than the three standard axioms of Lean (\texttt{propext}, \texttt{Classical.choice}, and \texttt{Quot.sound}). The formalization was written predominantly by Claude Opus 5.5 (Anthropic).

Therefore, beyond Lean and Mathlib, a reader has to trust only that the formal definitions in \texttt{ADHO/Defs} and the statements of the theorems in \texttt{ADHO.Main} say what the article says. Only these theorems are claimed to formalize results of the article in full: the definitions formalize those of the article only as far as the proofs need them, and \texttt{ADHO.Leaf} restates the remaining theorems on which no other result depends, some of which formalize only parts of results. The formal definitions and statements also depart from the article in ways that do not affect the results: $\mathbb{R}^2$ and $\mathbb{R}^3$ carry the norm $\lVert \cdot \rVert_{\infty}$ instead of the equivalent Euclidean norm; solutions on intervals are functions on $\mathbb{R}$, and the definitions do not presuppose that solutions are unique; and $f_1(0) = f_2(0) = g_1(0) = g_2(0) = 0$, $c(x, z, 0) = 0$, and the continuity of $h_1^{-1}$ and $h_2^{-1}$ are proved rather than assumed. In a semantic review recorded in LEAN-CHECK,\footnote{See \url{https://gitlab.com/user9716869/LEAN-CHECK}.} the AI models Claude Opus 5.5 and Claude Fable 5.2 have compared the formal definitions and statements with this article. The review is informal and provides no guarantee that the formalization says what the article says.

\section{Global Existence, Uniqueness, Boundedness}\label{sec:convergence}

Define 
\begin{equation}
\mathcal{E} \triangleq \left\{ (x, z, v) \in \mathbb{R}^3 : h_1(x) = -h_2(z) \wedge v = 0 \right\}
\end{equation}
Then, 
\begin{proposition}
$\mathcal{E}$ is the set of all equilibrium points of the system described by Eq. \eqref{eq:main}.
\end{proposition}
\begin{proof}
Suppose that $(x_e, z_e, v_e) \in \mathcal{E}$. Note that $v_e = 0$. Then, since $g_1(v_e) = g_2(v_e) = 0$, $\dot{z} = 0$ and $\dot{x} = 0$. Furthermore, $c(x_e, z_e, v_e) = 0$ and $-h_1(x_e) - h_2(z_e) = 0$. Therefore, $\dot{v} = 0$. Thus, $(x_e, z_e, v_e)$ is an equilibrium point of the system described by Eq. \eqref{eq:main}. Suppose that $(x_e, z_e, v_e)$ is an equilibrium point of the system described by Eq. \eqref{eq:main}. Then, since $\dot{x} = 0$, $v_e = 0$. Thus, $c(x_e, z_e, v_e) = 0$. Therefore, since $\dot{v} = 0$, $-h_1(x_e) - h_2(z_e) = 0$. Hence, $(x_e, z_e, v_e) \in \mathcal{E}$.
\end{proof}

Consider the following Lyapunov function candidate:
\begin{equation}
\mathcal{V} (x, z, v) = \int_0^x h_1(s) ds + \int_0^z h_2(s) ds + \frac{1}{2} v^2
\end{equation}
The following lemmas show that $\mathcal{V}$ is positive definite and radially unbounded:
\begin{lemma}\label{thm:pd}
$0 \leq \mathcal{V}(x, z, v)$ for all $x,z,v \in \mathbb{R}$; $\mathcal{V}(x, z, v) = 0$ if and only if $x = z = v = 0$. 
\end{lemma}
\begin{proof}
Note that $0 \leq \int_0^u h_1(s) ds$, $0 \leq \int_0^u h_2(s) ds$, and $0 \leq (1/2)u^2$ for all $u \in \mathbb{R}$. Thus, $0 \leq \mathcal{V}(x, z, v)$. It can be verified by substitution that $\mathcal{V}(0) = 0$. Suppose that $\mathcal{V}(x, z, v) = 0$. Note that $\int_0^u h_1(s) ds > 0$ and $\int_0^u h_2(s) ds > 0$ for all $u \in \mathbb{R} \setminus \{ 0 \}$. Thus, $\int_0^x h_1(s) ds = 0$, $\int_0^z h_2(s) ds = 0$ and $\frac{1}{2} v^2 = 0$. Therefore, $x = z = v = 0$. 
\end{proof}
\begin{lemma}\label{thm:radubd}
$\mathcal{V}$ is radially unbounded, that is, $\mathcal{V}(x,z,v) \rightarrow +\infty$ as $\lVert (x,z,v) \rVert_{\infty} \rightarrow +\infty$. 
\end{lemma}
\begin{proof}
Since $h_1$ and $h_2$ are strictly increasing and continuous, $\lim_{x \rightarrow \pm \infty} \int_0^x h_1(s) ds = +\infty$ and $\lim_{z \rightarrow \pm \infty} \int_0^z h_2(s) ds = +\infty$. Note also that $\lim_{v \rightarrow \pm \infty} \frac{1}{2} v^2 = +\infty$. Fix $M \in \mathbb{R}_{>0}$. Then, obtain $\delta_1, \delta_2, \delta_3 \in \mathbb{R}_{>0}$ such that $\abs{x} > \delta_1$ implies $\int_0^x h_1(s) ds > M$, $\abs{z} > \delta_2$ implies $\int_0^z h_2(s) ds > M$, and $\abs{v} > \delta_3$ implies $\frac{1}{2} v^2 > M$. Define $\delta \triangleq \max (\delta_1, \delta_2, \delta_3)$ and suppose that $\lVert (x, z, v) \rVert_{\infty} > \delta$. Then, either $\abs{x} > \delta_1$ or   $\abs{z} > \delta_2$ or $\abs{v} > \delta_3$. Thus, $\mathcal{V}(x, z, v) > M$. 
\end{proof}

It can be verified that 
\begin{equation}
\dot{\mathcal{V}}(x, z, v) = h_2(z) f_1(z) g_1(v) + h_2(z) f_2(z) g_2(v) - v c(x, z, v)
\end{equation}

\begin{lemma}\label{thm:dV_leq_0}
$\dot{\mathcal{V}}(x, z, v) \leq 0$ for all $x, z, v \in \mathbb{R}$.
\end{lemma}
\begin{proof}
Fix $x, z, v \in \mathbb{R}$. Note that $-v c(x, z, v) \leq 0$. Note also that $h_2(z) f_1(z) \leq 0$ and $0 \leq h_2(z) f_2(z)$ while $0 \leq g_1(v)$ and $g_2(v) \leq 0$. Thus, $h_2(z) f_1(z) g_1(v) \leq 0$ and $h_2(z) f_2(z) g_2(v) \leq 0$. Thus, $\dot{\mathcal{V}}(x, z, v) \leq 0$ as a sum of nonpositive terms. 
\end{proof}

\begin{proposition}\label{thm:eqb}
The solutions of the system given by Eq. \eqref{eq:main} are equibounded.
\end{proposition}
\begin{proof}
Noting that $\mathcal{V}$ is continuously differentiable, the proof follows from Lemmas \ref{thm:pd}, \ref{thm:radubd}, and \ref{thm:dV_leq_0} by, for example, Theorem 8.7 in Ref. \cite{yoshizawa_stability_1975}.
\end{proof}
\begin{remark}
See, e.g., Ref. \cite{kellett_compendium_2014} for a description of a relationship between the radially unbounded positive definite functions and class $\mathcal{K}_{\infty}$ functions.
\end{remark}

\begin{theorem}\label{thm:existence}
There exists a unique solution of the IVP given by Eq. \eqref{eq:main} for every initial condition with $x_0, z_0, v_0 \in \mathbb{R}$ on any time interval $[0, T)$ with $T \in \mathbb{R}_{>0} \cup \left\{ + \infty \right\}$. 
\end{theorem}
\begin{proof}
Taking into account that the state function of the system described by Eq. \eqref{eq:main} is locally Lipschitz continuous, the solutions are unique and exist on a non-empty maximal interval of existence (e.g., see Theorem 54 in Ref. \cite{sontag_mathematical_1998}). Therefore, taking into account Proposition \ref{thm:eqb}, by the theorem on the extendability of the solutions (e.g., see Proposition C.3.6 in Ref. \cite{sontag_mathematical_1998}), each solution can be extended to a unique solution on $[0, +\infty)$.
\end{proof}

\section{Stability and Convergence}\label{sec:SC}

\begin{proposition}\label{thm:LS0}
The origin $0 \in \mathbb{R}^{3}$ is a Lyapunov stable equilibrium point of the system given by Eq. \eqref{eq:main}. 
\end{proposition}
\begin{proof}
That the origin is an equilibrium point can be verified by substitution into Eq. \eqref{eq:main}. Noting that $\mathcal{V}$ is continuously differentiable, Lyapunov stability of the origin follows from Lemmas \ref{thm:pd}, \ref{thm:radubd}, and \ref{thm:dV_leq_0} by, for example, Theorem 4.1 in Ref. \cite{khalil_nonlinear_2002}. 
\end{proof}

Define 
\begin{equation}
\mathcal{M}_z \triangleq \left\{ (x, z, v) \in \mathbb{R}^3 : z = 0 \right\}
\end{equation}
\begin{equation}
\mathcal{M}_v \triangleq \left\{ (x, z, v) \in \mathbb{R}^3 : v = 0 \right\}
\end{equation}
Considering the proof of Lemma \ref{thm:dV_leq_0}, $\dot{\mathcal{V}}(x,z,v) = 0$ if and only if $h_2(z) f_1(z) g_1(v) = 0$,  $h_2(z) f_2(z) g_2(v) = 0$, and $v c(x, z, v) = 0$. The following lemma can be inferred immediately:
\begin{lemma}\label{thm:E_iV_Mz_Mv}
$\mathcal{E} \subseteq \mathcal{M}_v \subseteq \dot{\mathcal{V}}^{-1}(0) \subseteq \mathcal{M}_z \cup \mathcal{M}_v$
\end{lemma}
\begin{lemma}\label{thm:LIS}
With reference to the system given by Eq. \eqref{eq:main}, the largest invariant set that is contained in $\dot{\mathcal{V}}^{-1}(0)$ is $\mathcal{E}$.
\end{lemma}
\begin{proof}
Due to Lemma \ref{thm:E_iV_Mz_Mv}, it suffices to show that the largest invariant set that is contained in $\mathcal{M}_z \cup \mathcal{M}_v$ is $\mathcal{E}$. Suppose that $X = (x, z, v) : \mathbb{R}_{\geq 0} \longrightarrow \mathcal{M}_z \cup \mathcal{M}_v$ is a solution of the IVP given by Eq. \eqref{eq:main}.

Suppose that there exists $t_s \in \mathbb{R}_{\geq 0}$ such that 
\[
X(t_s) = (x(t_s), z(t_s), v(t_s)) \in \mathcal{M}_z \setminus (\mathcal{M}_z \cap \mathcal{M}_v)
\]
Then, $z(t_s) = 0$ and $v(t_s) \neq 0$. By continuity of $X$, obtain $t_e \in \mathbb{R}_{>t_s}$ such that $X(t) \in \mathcal{M}_z \setminus (\mathcal{M}_z \cap \mathcal{M}_v)$ for all $t \in [t_s, t_e)$. In this case, $z(t) = \dot{z}(t) = 0$ for all $t \in [t_s, t_e)$. Therefore, $v(t) = 0$ for all $t \in [t_s, t_e)$. However, in this case, $X(t) \in \mathcal{M}_z \cap \mathcal{M}_v$ for all $t \in [t_s, t_e)$, which results in a contradiction. Thus, $X(t) \in \mathcal{M}_v$ for all $t \in \mathbb{R}_{\geq 0}$. Therefore, $v(t) = 0$ for all $t \in \mathbb{R}_{\geq 0}$. Thus, $\dot{v}(t) = 0$ for all $t \in \mathbb{R}_{\geq 0}$. Therefore, $h_1(x(t)) + h_2(z(t)) = 0$ for all $t \in \mathbb{R}_{\geq 0}$. Thus, $X(t) \in \mathcal{E} \subseteq \mathcal{M}_z \cup \mathcal{M}_v$ for all $t \in \mathbb{R}_{\geq 0}$ or $X : \mathbb{R}_{\geq 0} \longrightarrow \mathcal{E}$. Since $\mathcal{E}$ is an invariant set, it is also the largest invariant set contained in $\mathcal{M}_z \cup \mathcal{M}_v$. 
\end{proof}

Define the set $\mathcal{S}_{X}$ as
\begin{equation}
\mathcal{S}_{X} \triangleq \{ Y \in \mathbb{R}^3 : \mathcal{V} (Y) \leq \mathcal{V} (X) \}
\end{equation}

\begin{proposition}\label{thm:convergence_E}
Suppose that $X : \mathbb{R}_{\geq 0} \longrightarrow \mathbb{R}^3$ is a solution of the IVP given by Eq. \eqref{eq:main} with $X_0 \triangleq X(0)$. Then, $\mathcal{O}_{X_0}^{+\infty}$ is a nonempty, compact, connected, invariant set such that $\lim_{t \rightarrow +\infty} X(t) = \mathcal{O}_{X_0}^{+\infty}$ and $\mathcal{O}_{X_0}^{+\infty} \subseteq \mathcal{S}_{X_0} \cap \mathcal{E}$.
\end{proposition}
\begin{proof}
Note that $\mathcal{O}_{X_0}^{+}$ is bounded by Proposition \ref{thm:eqb}. Thus, $\mathcal{O}_{X_0}^{+\infty}$ is a nonempty, compact, connected, invariant set, and $\lim_{t \rightarrow +\infty} X(t) = \mathcal{O}_{X_0}^{+\infty}$ (e.g., see Proposition 5.1 in Ref. \cite{bhat_nontangency-based_2003} or Theorem 2.41 in Ref. \cite{haddad_nonlinear_2011}). Note that, by Lemma \ref{thm:LIS}, the largest invariant set contained in $\dot{\mathcal{V}}^{-1}(0)$ is $\mathcal{E}$. Suppose that $\mathcal{P}$ is the largest invariant set contained in $\mathcal{S}_{X_0}$. Then, $\mathcal{O}_{X_0}^{+\infty} \subseteq \mathcal{P} \cap \mathcal{E} \subseteq \mathcal{S}_{X_0} \cap \mathcal{E}$ (e.g., see Proposition 5.3 in Ref. \cite{bhat_nontangency-based_2003}).
\end{proof}
Define the set $\mathcal{E}_{xz}$ as
\begin{equation}
\mathcal{E}_{xz} \triangleq \{ (x, z) \in \mathbb{R}^2 : h_1(x) = -h_2(z) \}
\end{equation}
The following corollary is an immediate consequence of Proposition \ref{thm:convergence_E}:
\begin{corollary}\label{thm:convergence_Exz_V0}
Suppose that $x, z, v : \mathbb{R}_{\geq 0} \longrightarrow \mathbb{R}$ form a solution of the IVP given by Eq. \eqref{eq:main}. Then, $\lim_{t \rightarrow +\infty} (x(t), z(t)) = \mathcal{E}_{xz}$ and $\lim_{t \rightarrow +\infty} v(t) = 0$.
\end{corollary}

\begin{lemma}\label{thm:HJz}
Suppose that $H : \mathbb{R} \longrightarrow \mathbb{R}$ is a strictly decreasing homeomorphism such that $H(0) = 0$. Define $J : \mathbb{R} \longrightarrow \mathbb{R}_{\geq 0}$ as $J(y) \triangleq d ((0, y), \mathcal{G}(H))$ for all $y \in \mathbb{R}$. Then, $J$ is a continuous function such that $J(0) = 0$ and $J(y) > 0$ for all $y \in \mathbb{R} \setminus \{ 0 \}$. Moreover, for every $y \in \mathbb{R}$, $J(y) = d((0, y), (a, b))$ for some $(a, b) \in \mathcal{G}(H)$ such that $a \leq 0$ and $0 \leq b$ if $0 \leq y$, and $0 \leq a$ and $b \leq 0$ if $y \leq 0$. Furthermore, $J \restriction \mathbb{R}_{\geq 0}$ is nondecreasing and $J \restriction \mathbb{R}_{\leq 0}$ is nonincreasing. 
\end{lemma}
\begin{proof}
Define the function $G : \mathbb{R} \times \mathbb{R} \longrightarrow \mathbb{R}_{\geq 0}$ as
\[
G(y, x) \triangleq d((0, y), (x, H(x))) = \sqrt{x^2 + (y - H(x))^2}
\]
and note that $J(y) = \inf_{x \in \mathbb{R}} G(y, x)$. $J$ is continuous due to continuity of $d(\cdot, S)$ for any non-empty set $S$. Since $H(0) = 0$, $(0, 0) \in \mathcal{G}(H)$, and, therefore, $J(0) = d((0, 0), \mathcal{G}(H)) = 0$. 

Suppose to the contrary that $J(y) = 0$ for some $y \neq 0$. Without the loss of generality assume that $y > 0$. Define 
\[
\varepsilon \triangleq \min \left( - H^{-1} \left( \frac{y}{2} \right), \frac{y}{2} \right)
\]
Since $J(y) = 0$, obtain $x \in \mathbb{R}$ such that $\sqrt{x^2 + (y - H(x))^2} < \varepsilon$. Then, $\abs{x} < - H^{-1} \left( y/2 \right)$ and $\abs{y - H(x)} < y/2$. Then, $H(x) < y/2$ follows from the former and $y/2 < H(x)$ follows from the latter. Thus, by contradiction, $J(y) \neq 0$.

For every $b \in \mathbb{R}_{>0}$, define $f_b^{+} : \mathbb{R}_{>0} \longrightarrow \mathbb{R}_{\geq 0}$ as
\[
f_b^{+} (y) \triangleq \min_{x \in [H^{-1}(b), 0]} G(y, x)
\]
for every $y \in \mathbb{R}_{>0}$. 

It will now be shown that $J(y) = f_y^{+}(y)$ for all $y \in \mathbb{R}_{>0}$. Fix $y \in \mathbb{R}_{>0}$. Note that $f_y^{+} (y) \leq y = G(y, 0)$. Suppose that $x \in \mathbb{R}_{>0}$. Then, since $H(x) < 0$,
\[
G(y, x) = \sqrt{x^2 + (y - H(x))^2} > \abs{y - H(x)} > y
\]
Therefore, 
\[
f_y^{+} (y) \leq y \leq \inf_{x \in (0, +\infty)} G(y, x)
\]
Note that $f_y^{+} (y) \leq \abs{H^{-1}(y)} = G(y, H^{-1}(y))$. Fix $x \in \mathbb{R}_{<H^{-1}(y)}$. Then,
\[
G(y, x) = \sqrt{x^2 + (y - H(x))^2} \geq \abs{x} > \abs{H^{-1}(y)}
\]
Thus, 
\[
f_y^{+} (y) \leq \abs{H^{-1}(y)} \leq \inf_{x \in \left(-\infty, H^{-1}(y) \right)} G(y, x)
\]
The sets 
\[
\{ G(y, x) : x \in [H^{-1}(y), 0] \}
\] 
\[
\{ G(y, x) : x \in (0, +\infty) \}
\] 
\[
\{ G(y, x) : x \in \left( -\infty, H^{-1}(y) \right) \}
\] 
are bounded from below by $0$, and their union is 
\[
\{ G(y, x) : x \in \mathbb{R} \}
\]
since $H^{-1}(y) < 0$. The infimum of the first set is $f_y^{+} (y)$, which, by the inequalities above, does not exceed the infima of the other two sets. Thus, by Lemma \ref{thm:inf_A_U_B}, applied first to the first two sets and then to their union and the third set, $J(y) = f_y^{+} (y)$. This also implies that 
\[
J(y) = f_b^{+}(y) = \min_{x \in [H^{-1}(b), 0]} (G \restriction (0, b) \times [H^{-1}(b), 0]) (y, x)
\]
for all $y \in (0, b)$ for all $b \in \mathbb{R}_{>0}$. Furthermore, this implies that, for every $y \in \mathbb{R}_{\geq 0}$, there exists $a \in [H^{-1}(y), 0]$ such that $J(y) = d((0, y), (a, H(a)))$ (with $a \leq 0$ and $0 \leq H(a)$). 

Note that $J(y) < G(y, x)$ for all $y \in \mathbb{R}_{>0}$ and $x \in \mathbb{R}_{<H^{-1}(y)}$. To show this, fix $y \in \mathbb{R}_{>0}$ and $x \in \mathbb{R}_{<H^{-1}(y)}$ and, by the properties of $H$, obtain $z \in \mathbb{R}_{>0}$ such that $y < z$ and $x = H^{-1}(z)$. Then, 
\[
J(y) \leq \abs{H^{-1}(y)} < \abs{H^{-1}(z)} < G(y, H^{-1}(z)) = G(y, x)
\]

Fix $b \in \mathbb{R}_{>0}$. Noting that $G \restriction (0, b) \times [H^{-1}(b), 0]$ is continuous and continuously differentiable with respect to the first argument, $(0, b)$ is open, $[H^{-1}(b), 0]$ is compact, by Danskin's Theorem (\cite{danskin_theory_1966}, see also Theorem 1.29 in Ref. \cite{guler_foundations_2010}), the right derivative of $J$ exists for every $y \in (0, b)$, and is given by
\[
\partial^{+} J(y) = \min_{x \in S(y)} \partial_1 G (y, x) =  \min_{x \in S(y)} \frac{y - H(x)}{\sqrt{x^2 + (y - H(x))^2}}
\]
where
\[
S(y) \triangleq \{ x \in [H^{-1}(b), 0] : J(y) = G(y, x) \}
\]
Fix $y \in (0, b)$. Taking into account that $J(y) < G(y, x)$ for all $x \in \mathbb{R}_{<H^{-1}(y)}$, $0 \leq y - H(x)$ for all $x \in S(y)$. Thus, $\partial^{+} J(y) \geq 0$ for all $y \in \mathbb{R}_{>0}$. Therefore, $J$ is nondecreasing on $\mathbb{R}_{>0}$. Since $J(0) = 0$ and $J(y) > 0$ for all $y \in \mathbb{R}_{>0}$, $J$ is nondecreasing on $\mathbb{R}_{\geq 0}$. 

By similar arguments, it can be shown that for every $y \in \mathbb{R}_{\leq 0}$, there exists $a \in [0, H^{-1}(y)]$ such that $J(y) = d((0, y), (a, H(a)))$ (with $0 \leq a$ and $H(a) \leq 0$), as well as that $J$ is nonincreasing on $\mathbb{R}_{\leq 0}$.
\end{proof}

By similar arguments, it is possible to prove the following lemma:

\begin{lemma}\label{thm:HJx}
Suppose that $H : \mathbb{R} \longrightarrow \mathbb{R}$ is a strictly decreasing homeomorphism such that $H(0) = 0$. Define $J : \mathbb{R} \longrightarrow \mathbb{R}_{\geq 0}$ as $J(x) \triangleq d ((x, 0), \mathcal{G}(H))$ for all $x \in \mathbb{R}$. Then, $J$ is a continuous function such that $J(0) = 0$ and $J(x) > 0$ for all $x \in \mathbb{R} \setminus \{ 0 \}$. Moreover, for every $x \in \mathbb{R}$, $J(x) = d((x, 0), (a, b))$ for some $(a, b) \in \mathcal{G}(H)$ such that $a \leq 0$ and $0 \leq b$ if $x \leq 0$, and $0 \leq a$ and $b \leq 0$ if $0 \leq x$. Furthermore, $J \restriction \mathbb{R}_{\geq 0}$ is nondecreasing and $J \restriction \mathbb{R}_{\leq 0}$ is nonincreasing. 
\end{lemma}

Define the sets
\begin{equation}
\mathcal{M}_x^{+} \triangleq \{ (x, z, v) \in \mathbb{R}^3 : z = 0 \wedge 0 \leq x \}
\end{equation}
\begin{equation}
\mathcal{M}_x^{-} \triangleq \{ (x, z, v) \in \mathbb{R}^3 : z = 0 \wedge x \leq 0 \}
\end{equation}
\begin{equation}
\mathcal{M}_z^{+} \triangleq \{ (x, z, v) \in \mathbb{R}^3 : x = 0 \wedge 0 \leq z \}
\end{equation}
\begin{equation}
\mathcal{M}_z^{-} \triangleq \{ (x, z, v) \in \mathbb{R}^3 : x = 0 \wedge z \leq 0  \}
\end{equation}
\begin{equation}
\mathcal{M}_{xz} \triangleq \mathcal{M}_x^{+} \cup \mathcal{M}_x^{-} \cup \mathcal{M}_z^{+} \cup \mathcal{M}_z^{-}
\end{equation}

\begin{lemma}\label{thm:Mz_frequently}
Suppose that $X = (x,z,v): \mathbb{R}_{\geq 0} \longrightarrow \mathbb{R}^3$ is a solution of the IVP given by Eq. \eqref{eq:main}. Suppose that $X$ is in $\mathcal{M}_z^{+}$ frequently, that is, for all $T \in \mathbb{R}_{>0}$ there exists $t \in \mathbb{R}_{>T}$ such that $X(t) \in \mathcal{M}_z^{+}$. Then, $\lim_{t \rightarrow +\infty} X(t) = 0$. 
\end{lemma}
\begin{proof}
Fix $\varepsilon \in \mathbb{R}_{>0}$. By Proposition \ref{thm:LS0}, obtain $\delta \in \mathbb{R}_{>0}$ such that if $X(T) \in \mathbb{B}(0, \delta)$ for some $T \in \mathbb{R}_{\geq 0}$, then $X(t) \in \mathbb{B}(0, \varepsilon)$ for all $t > T$. Define $H : \mathbb{R} \longrightarrow \mathbb{R}$ as $H(x) \triangleq h_2^{-1} (- h_1 (x))$ for all $x \in \mathbb{R}$. Note that $H$ is a strictly decreasing homeomorphism such that $H(0) = 0$. Note also that $\mathcal{E}_{xz} = \mathcal{G}(H)$. Define $J : \mathbb{R} \longrightarrow \mathbb{R}_{\geq 0}$ as $J(z) \triangleq d((0, z), \mathcal{E}_{xz})$. 

Using Corollary \ref{thm:convergence_Exz_V0}, obtain $T \in \mathbb{R}_{\geq 0}$ such that $\abs{v(t)} < \delta/\sqrt{2}$ and $d((x(t), z(t)), \mathcal{E}_{xz}) < J \left( \delta/\sqrt{2} \right)$ for all $t \in \mathbb{R}_{>T}$. Using the assumptions of the lemma, obtain $t \in \mathbb{R}_{>T}$ such that $X(t) \in \mathcal{M}_z^{+}$. Note that $x(t) = 0$, $0 \leq z(t)$, and $\abs{v(t)} < \delta/\sqrt{2}$. Therefore, 
\[
d((x(t), z(t)), \mathcal{E}_{xz}) = d((0, z(t)), \mathcal{E}_{xz}) = J(z(t)) < J \left( \delta/\sqrt{2} \right)
\]
Thus, $z(t) \leq \delta/\sqrt{2}$ by Lemma \ref{thm:HJz}. Therefore, 
\[
\lVert X(t) \rVert_2 = \sqrt{x(t)^2 + z(t)^2 + v(t)^2} < \sqrt{\frac{\delta^2}{2} + \frac{\delta^2}{2}} = \delta
\]
Thus, $X(t) \in \mathbb{B}(0, \delta)$. Therefore, $X(s) \in \mathbb{B}(0, \varepsilon)$ for all $s \in \mathbb{R}_{>t}$. By generalization, for every $\varepsilon \in \mathbb{R}_{>0}$, there is $T \in \mathbb{R}_{\geq 0}$ such that $X(t) \in \mathbb{B}(0, \varepsilon)$ for all $t > T$. Therefore, $\lim_{t \rightarrow +\infty} X(t) = 0$. 
\end{proof}
The argument that was used in the proof of Lemma \ref{thm:Mz_frequently} can be applied to other half-planes:
\begin{lemma}\label{thm:Mij_frequently}
Suppose that $X : \mathbb{R}_{\geq 0} \longrightarrow \mathbb{R}^3$ is a solution of the IVP given by Eq. \eqref{eq:main}. Suppose that $i \in \{ x, z \}$ and $j \in \{ -, + \}$. Suppose also that $X$ is in $\mathcal{M}_i^j$ frequently. Then, $\lim_{t \rightarrow +\infty} X(t) = 0$. 
\end{lemma}

Define
\begin{equation}
\mathcal{N}_{ne} \triangleq \{ (x, z, v) \in \mathbb{R}^3 : x > 0 \wedge z > 0 \}
\end{equation}
\begin{equation}
\mathcal{N}_{nw} \triangleq \{ (x, z, v) \in \mathbb{R}^3 : x < 0 \wedge z > 0 \}
\end{equation}
\begin{equation}
\mathcal{N}_{sw} \triangleq \{ (x, z, v) \in \mathbb{R}^3 : x < 0 \wedge z < 0 \}
\end{equation}
\begin{equation}
\mathcal{N}_{se} \triangleq \{ (x, z, v) \in \mathbb{R}^3 : x > 0 \wedge z < 0 \}
\end{equation}
\begin{equation}
\mathcal{N}_{xz} \triangleq \mathcal{N}_{ne} \cup \mathcal{N}_{nw} \cup \mathcal{N}_{sw} \cup \mathcal{N}_{se}
\end{equation}

\begin{lemma}\label{thm:N_ne_nw_sw_se}
Suppose that $X : \mathbb{R}_{\geq 0} \longrightarrow \mathbb{R}^3$ is a solution of the IVP given by Eq. \eqref{eq:main}. Suppose also that $X$ is not in $\mathcal{M}_{xz}$ eventually, that is, there exists $T \in \mathbb{R}_{\geq 0}$ such that $X(t) \not\in \mathcal{M}_{xz}$ for all $t \in \mathbb{R}_{>T}$. Then, there exists $i \in \{ ne, nw, sw, se \}$ such that $X$ is in $\mathcal{N}_i$ eventually.
\end{lemma}
\begin{proof}
Obtain $T \in \mathbb{R}_{>0}$ such that $X(t) \not \in \mathcal{M}_{xz}$ for all $t \in \mathbb{R}_{>T}$. Therefore, $X(t) \in \mathbb{R}^3 \setminus \mathcal{M}_{xz} = \mathcal{N}_{xz}$ for all $t \in \mathbb{R}_{>T}$. Note that $X(\mathbb{R}_{>T})$ is a connected set. Note also that $\mathcal{N}_{ne}$, $\mathcal{N}_{nw}$, $\mathcal{N}_{sw}$, and $\mathcal{N}_{se}$ are pairwise separated. Thus, by Lemma \ref{thm:con_sep}, there is $i \in \{ ne, nw, sw, se \}$ such that $X(t) \in \mathcal{N}_i$ for all $t \in \mathbb{R}_{>T}$.
\end{proof}

\begin{lemma}\label{thm:N_ne_sw}
Suppose that $X : \mathbb{R}_{\geq 0} \longrightarrow \mathbb{R}^3$ is a solution of the IVP given by Eq. \eqref{eq:main}. Suppose also that $X$ is in $\mathcal{N}_{ne}$ eventually or $X$ is in $\mathcal{N}_{sw}$ eventually. Then, $\lim_{t \rightarrow +\infty} X(t) = 0$. 
\end{lemma}
\begin{proof}
Fix $\varepsilon \in \mathbb{R}_{>0}$. By Proposition \ref{thm:LS0}, obtain $\delta \in \mathbb{R}_{>0}$ such that if $X(T) \in \mathbb{B}_{\infty}(0, \delta)$ for some $T \in \mathbb{R}_{\geq 0}$, then $X(t) \in \mathbb{B}(0, \varepsilon)$ for all $t > T$. Suppose that $X$ is in $\mathcal{N}_{ne}$ eventually. Obtain $T_1 \in \mathbb{R}_{\geq 0}$ such that $X(t) \in \mathcal{N}_{ne}$ for all $t \in \mathbb{R}_{>T_1}$. Define $H : \mathbb{R} \longrightarrow \mathbb{R}$ as $H(x) \triangleq h_2^{-1} (- h_1 (x))$ for all $x \in \mathbb{R}$. Note that $H$ is a strictly decreasing homeomorphism such that $H(0) = 0$. Note also that $\mathcal{E}_{xz} = \mathcal{G}(H)$. Define $J_1 : \mathbb{R} \longrightarrow \mathbb{R}_{\geq 0}$ as $J_1(z) \triangleq d((0, z), \mathcal{E}_{xz})$ and $J_2 : \mathbb{R} \longrightarrow \mathbb{R}_{\geq 0}$ as $J_2(x) \triangleq d((x, 0), \mathcal{E}_{xz})$. 

Using Corollary \ref{thm:convergence_Exz_V0}, obtain $T_2 \in \mathbb{R}_{\geq 0}$ such that $\abs{v(t)} < \delta$ and $d((x(t), z(t)), \mathcal{E}_{xz}) < \min ( \delta, J_1 \left( \delta/2 \right), J_2 \left( \delta/2 \right) )$ for all $t \in \mathbb{R}_{>T_2}$. Define $T \in \mathbb{R}_{\geq 0}$ as $T \triangleq \max (T_1, T_2)$. Fix $t \in \mathbb{R}_{>T}$. Introduce the notation 
\[
\mathcal{E}_{xz}^{+} \triangleq \{ (x, H(x)) : x \in \mathbb{R}_{\leq 0} \} 
\]
\[
\mathcal{E}_{xz}^{-} \triangleq \{ (x, H(x)) : x \in \mathbb{R}_{\geq 0} \}
\]
and note that $\mathcal{E}_{xz} = \mathcal{E}_{xz}^{+} \cup \mathcal{E}_{xz}^{-}$. Then, by Lemma \ref{thm:inf_A_U_B}, either 
\[
d((x(t), z(t)), \mathcal{E}_{xz}) = d((x(t), z(t)), \mathcal{E}_{xz}^{+})
\] 
or 
\[
d((x(t), z(t)), \mathcal{E}_{xz}) = d((x(t), z(t)), \mathcal{E}_{xz}^{-})
\]
Suppose that $\delta \leq z(t)$. Suppose also that 
\[
d((x(t), z(t)), \mathcal{E}_{xz}) = d((x(t), z(t)), \mathcal{E}_{xz}^{-})
\]
Then,
\[
\begin{aligned}
d((x(t), z(t)), \mathcal{E}_{xz}) & = \inf_{(a, b) \in \mathcal{E}_{xz}^{-}} \sqrt{(x(t) - a)^2 + (z(t) - b)^2}\\
& \geq \inf_{b \in \mathbb{R}_{\leq 0}} \abs{z(t) - b} \geq z(t) \geq \delta \\
& > d((x(t), z(t)), \mathcal{E}_{xz})
\end{aligned}
\]
Therefore, by contradiction, 
\[
d((x(t), z(t)), \mathcal{E}_{xz}) = d((x(t), z(t)), \mathcal{E}_{xz}^{+})
\]
Noting that 
\[
d((x(t), z(t)), \mathcal{E}_{xz}^{+}) \geq d((0, z(t)), \mathcal{E}_{xz}^{+}) = d((0, z(t)), \mathcal{E}_{xz})
\]
and $d((x(t), z(t)), \mathcal{E}_{xz}) < J_1(\delta/2)$, $d((0, z(t)), \mathcal{E}_{xz}) < J_1(\delta/2)$. Then, by Lemma \ref{thm:HJz}, $\delta \leq z(t) \leq \delta/2 < \delta$, which results in a contradiction. Thus, $z(t) < \delta$. By a similar argument, $x(t) < \delta$. Recalling that $\abs{v(t)} < \delta$, $X(t) \in \mathbb{B}_{\infty}(0, \delta)$. Therefore, for all $s \in \mathbb{R}_{>t}$, $X(s) \in \mathbb{B}(0, \varepsilon)$. By generalization, for every $\varepsilon \in \mathbb{R}_{>0}$ there is $T \in \mathbb{R}_{\geq 0}$ such that $X(t) \in \mathbb{B}(0, \varepsilon)$ for all $t > T$. Thence, $\lim_{t \rightarrow +\infty} X(t) = 0$.

The proof follows by a similar argument if $X$ is in $\mathcal{N}_{sw}$ eventually.
\end{proof}

Define $\mathcal{W} : \mathbb{R}^3 \longrightarrow \mathbb{R}$ as
\begin{equation}
\mathcal{W}(x, z, v) \triangleq (z - x)^2
\end{equation}
for all $x, z, v \in \mathbb{R}$. Note that 
\begin{equation}
\dot{\mathcal{W}}(x, z, v) = 2 (z - x)(f_1(z) g_1(v) + f_2(z) g_2(v))
\end{equation}
for all $x, z, v \in \mathbb{R}$. Define 
\begin{equation}
\mathcal{A}(L) \triangleq \{ (x, z, v) \in \mathbb{R}^3 : z = x + L \wedge v = 0 \}
\end{equation}
for all $L \in \mathbb{R}$.

\begin{lemma}\label{thm:A_convergence}
Suppose that $X : \mathbb{R}_{\geq 0} \longrightarrow \mathbb{R}^3$ is a solution of the IVP given by Eq. \eqref{eq:main}. If $X$ is in $\mathcal{N}_{nw}$ eventually, then there exists $L \in \mathbb{R}_{\geq 0}$ such that $\lim_{t \rightarrow +\infty} X(t) = \mathcal{A}(L)$. If $X$ is in $\mathcal{N}_{se}$ eventually, then there exists $L \in \mathbb{R}_{\leq 0}$ such that $\lim_{t \rightarrow +\infty} X(t) = \mathcal{A}(L)$. 
\end{lemma}
\begin{proof}
Suppose that $X$ is in $\mathcal{N}_{nw}$ eventually. Obtain $T_1 \in \mathbb{R}_{>0}$ such that $X(t) \in \mathcal{N}_{nw}$ for all $t \in \mathbb{R}_{>T_1}$. Note that $\dot{\mathcal{W}}(X) \leq 0$ for all $X \in \mathcal{N}_{nw}$. Thus, $\mathcal{W} \circ (X \restriction \mathbb{R}_{>T_1})$ is nonincreasing. Since $\mathcal{W}$ is bounded from below by $0 \in \mathbb{R}$, obtain $L \in \mathbb{R}_{\geq 0}$ such that $\lim_{t \rightarrow +\infty} \mathcal{W}(X(t)) = L^2$. Then, $\lim_{t \rightarrow +\infty} (z(t) - x(t)) = L$. Fix $\varepsilon \in \mathbb{R}_{>0}$. Taking into account Corollary \ref{thm:convergence_Exz_V0}, obtain $T_2 \in \mathbb{R}_{>T_1}$ such that $\abs{z(t) - x(t) - L} < \varepsilon$ and $\abs{v(t)} < \varepsilon/\sqrt{2}$ for all $t \in \mathbb{R}_{>T_2}$. Fix $t \in \mathbb{R}_{>T_2}$. Then,
\[
\begin{aligned}
d((x(t), z(t), v(t)), \mathcal{A}(L)) & = \sqrt{\frac{1}{2}\abs{z(t) - x(t) - L}^2 + \abs{v(t)}^2} \\
& < \sqrt{\frac{1}{2} \varepsilon^2 + \frac{1}{2}\varepsilon^2} = \varepsilon
\end{aligned}
\]
By generalization, for all $\varepsilon \in \mathbb{R}_{>0}$, there exists $T \in \mathbb{R}_{>0}$ such that $d((x(t), z(t), v(t)), \mathcal{A}(L)) < \varepsilon$. Thus, $\lim_{t \rightarrow +\infty} X(t) = \mathcal{A}(L)$.

The proof is similar if $X$ is in $\mathcal{N}_{se}$ eventually.
\end{proof}

\begin{lemma}\label{thm:N_nw_se}
Suppose that $X : \mathbb{R}_{\geq 0} \longrightarrow \mathbb{R}^3$ is a solution of the IVP given by Eq. \eqref{eq:main} with $X_0 \triangleq X(0)$. Suppose also that $X$ is in $\mathcal{N}_{nw}$ eventually or $X$ is in $\mathcal{N}_{se}$ eventually. Then, there exists $X^{*} \in \mathcal{E} \cap \mathcal{S}_{X_0}$ such that $\lim_{t \rightarrow +\infty} X(t) = X^{*}$. 
\end{lemma}
\begin{proof}
Define $H : \mathbb{R} \longrightarrow \mathbb{R}$ as $H(x) \triangleq h_2^{-1} (- h_1 (x))$ for all $x \in \mathbb{R}$. Note that $H$ is a strictly decreasing homeomorphism such that $H(0) = 0$. 

Suppose that $X$ is in $\mathcal{N}_{nw}$ eventually. Obtain $T_1 \in \mathbb{R}_{>0}$ such that $X(t) \in \mathcal{N}_{nw}$ for all $t \in \mathbb{R}_{>T_1}$. By Lemma \ref{thm:A_convergence}, obtain $L \in \mathbb{R}_{\geq 0}$ such that $\lim_{t \rightarrow +\infty} X(t) = \mathcal{A}(L)$. 

By Lemma \ref{thm:fgx}, obtain the unique $a \in \mathbb{R}_{\leq 0}$ such that
\[
X^{*} \triangleq (a, H(a), 0) = (a, a + L, 0)
\]
Then, $\mathcal{E} \cap \mathcal{A}(L) = \{ X^{*} \}$.

Fix $p \in \mathcal{O}_{X_0}^{+\infty}$. Fix $\varepsilon \in \mathbb{R}_{>0}$. Note that $\mathcal{O}_{X_0}^{+\infty} \subseteq \mathcal{E} \cap \text{cl} \: \mathcal{N}_{nw}$ with $\mathcal{O}_{X_0}^{+\infty}$ being a compact set by Proposition \ref{thm:convergence_E}. Obtain $T_2 \in \mathbb{R}_{>T_1}$ such that $d(X(t), \mathcal{A}(L)) < \varepsilon/2$ for all $t \in \mathbb{R}_{>T_2}$. Obtain $t \in \mathbb{R}_{>T_2}$ such that $d(p, X(t)) < \varepsilon/2$. Then, by the triangle inequality,
\[
d(p, \mathcal{A}(L)) \leq d(p, X(t)) + d(X(t), \mathcal{A}(L)) < \frac{\varepsilon}{2} + \frac{\varepsilon}{2} = \varepsilon
\]
Thus, by generalization, $d(p, \mathcal{A}(L)) < \varepsilon$ for all $\varepsilon \in \mathbb{R}_{> 0}$. Therefore, by analysis, $d(p, \mathcal{A}(L)) = 0$. Since $\mathcal{A}(L)$ is a closed subset of $\mathbb{R}^3$, $p \in \mathcal{A}(L)$. Thus, by generalization, $\mathcal{O}_{X_0}^{+\infty} \subseteq \mathcal{A}(L)$. Thence, 
\[
\mathcal{O}_{X_0}^{+\infty} = \mathcal{O}_{X_0}^{+\infty} \cap \mathcal{A}(L) \subseteq \mathcal{E} \cap \mathcal{A}(L) =  \{ X^{*} \}
\]
Therefore, $\lim_{t \rightarrow +\infty} X(t) = X^{*}$. 

The proof is similar if $X$ is in $\mathcal{N}_{se}$ eventually.
\end{proof}

\begin{theorem}\label{thm:main}
Suppose that $X : \mathbb{R}_{\geq 0} \longrightarrow \mathbb{R}^3$ is a solution of the IVP given by Eq. \eqref{eq:main} with $X_0 \triangleq X(0)$. Then, there exists $X^{*} \in \mathcal{E} \cap \mathcal{S}_{X_0}$ such that $\lim_{t \rightarrow +\infty} X(t) = X^{*}$. 
\end{theorem}
\begin{proof}
Suppose that $X$ is in $\mathcal{M}_{xz}$ frequently. Then, $\lim_{t \rightarrow +\infty} X(t) = 0 \in \mathcal{E} \cap \mathcal{S}_{X_0}$ by Lemma \ref{thm:Mij_frequently}. Suppose that $X$ is not in $\mathcal{M}_{xz}$ eventually. Then, by Lemma \ref{thm:N_ne_nw_sw_se}, there exists $i \in \{ ne, nw, sw, se \}$ such that $X$ is in $\mathcal{N}_i$ eventually. Suppose that $X$ is in $\mathcal{N}_{ne}$ eventually or $X$ is in $\mathcal{N}_{sw}$ eventually. Then, by Lemma \ref{thm:N_ne_sw}, $\lim_{t \rightarrow +\infty} X(t) = 0 \in \mathcal{E} \cap \mathcal{S}_{X_0}$. Suppose that $X$ is in $\mathcal{N}_{nw}$ eventually or $X$ is in $\mathcal{N}_{se}$ eventually. Then, by Lemma \ref{thm:N_nw_se}, there exists $X^{*} \in \mathcal{E} \cap \mathcal{S}_{X_0}$ such that $\lim_{t \rightarrow +\infty} X(t) = X^{*}$. 
\end{proof}

\section{The Bouc-Wen Oscillator}\label{sec:BW}

In the context of plasticity, arguably, the most well-known specialization of the Duhem models is the Bouc-Wen model \cite{bouc_forced_1968, bouc_modemathematique_1971, wen_method_1976}. The Bouc-Wen model is a general parameterizable rate-independent differential model of hysteresis. It can be described by the following system of differential equations with an input and an output \cite{ikhouane_systems_2007}
\begin{equation}
\begin{cases}
\dot{x} = v \\
\dot{z} = D^{-1} A v - D^{-1} \beta \abs{z}^{n - 1} z \abs{v} - D^{-1} \gamma \abs{z}^n v \\
F = -\alpha k x - (1 - \alpha) D k z \\
\begin{matrix} x(0) = x_0 & z(0) = z_0  \end{matrix}
\end{cases}
\end{equation}
where $x, z \in \mathbb{R}$ are state variables, $v \in \mathbb{R}$ is an input variable, $F \in \mathbb{R}$ is an output variable, $A, \beta, \gamma \in \mathbb{R}$, $\alpha \in (0, 1) \subseteq \mathbb{R}$, $k \in \mathbb{R}_{>0}$, $D \in \mathbb{R}_{>0}$, $n \in \mathbb{R}_{>1}$, $\beta \neq - \gamma$, $x_0, z_0 \in \mathbb{R}$ are parameters. Recent surveys in Refs. \cite{ikhouane_systems_2007, ismail_hysteresis_2009, heredia-perez_state---art_2025} provide further general information about the model.

According to Ref. \cite{ikhouane_dynamic_2007}, the Bouc-Wen model is of class I if and only if $A \in \mathbb{R}_{>0}$, $\gamma \in (-\beta, \beta]$. The authors of Ref. \cite{ikhouane_dynamic_2007} describe four other distinct classes of the Bouc-Wen model that depend on the model parameters. However, these classes are seldom relevant for the applications related to mechanical oscillations: ``class I is the only one that is BIBO stable, is compatible with the free motion of the real systems described by the Bouc–Wen model, is passive and is compatible with the laws of thermodynamics'' \cite{ikhouane_dynamic_2007}. 

In Ref. \cite{ikhouane_dynamic_2007}, the authors study a model of a point mass attached to a rigid surface via a viscous damper and a Bouc-Wen elastoplastic element. The dynamics of the system are given by the following IVP:
\begin{equation}\label{eq:BWO_original}
\begin{cases}
\dot{x} = v \\
\dot{z} = D^{-1} A v - D^{-1} \beta \abs{z}^{n - 1} z \abs{v} - D^{-1} \gamma \abs{z}^n v \\
\dot{v} = -\alpha \frac{k}{m} x - (1 - \alpha) D \frac{k}{m} z - \frac{b}{m} v \\
\begin{matrix} x(0) = x_0 & z(0) = z_0 & v(0) = v_0 \end{matrix}
\end{cases}
\end{equation}
Here, $x, z, v \in \mathbb{R}$ are state variables, $m \in \mathbb{R}_{>0}$, $b \in \mathbb{R}_{\geq 0}$, and $v_0 \in \mathbb{R}$ are additional model parameters. In Ref. \cite{ikhouane_dynamic_2007}, it is shown that if the Bouc-Wen model belongs to class I or class II, then there exists $\bar{b} \in \mathbb{R}_{\geq 0}$ such that $\bar{b} \leq b$ implies that $x_{\infty}, z_{\infty} \in \mathbb{R}$ such that $\lim_{t \rightarrow +\infty} x(t) = x_{\infty}$, $\lim_{t \rightarrow +\infty} z(t) = z_{\infty}$, $\lim_{t \rightarrow +\infty} v(t) = 0$, and $\alpha x_{\infty} + (1 - \alpha) D  z_{\infty} = 0$. 

Under the conditions stated below, the Bouc-Wen oscillator is a special case of the Duhem oscillator that was considered in this study. To show this, suppose that $A \in \mathbb{R}_{>0}$ and introduce a new state variable $Z \triangleq z/(A D^{-1})$. Then, after rescaling and renaming (i.e., replacing the symbol $Z$ with $z$), Eq. \eqref{eq:BWO_original} becomes
\begin{equation}\label{eq:BWO_rescaled}
\begin{cases}
\dot{x} = v \\
\dot{z} = v - A^{n - 1} D^{-n} \beta \abs{z}^{n - 1} z \abs{v} - A^{n - 1} D^{-n} \gamma \abs{z}^n v \\
\dot{v} = -\alpha \frac{k}{m} x - (1 - \alpha) A \frac{k}{m} z - \frac{b}{m} v \\
\begin{matrix} x(0) = x_0 & z(0) = z_0/(A D^{-1}) & v(0) = v_0 \end{matrix}
\end{cases}
\end{equation}
Define $f_1, f_2, g_1, g_2, h_1, h_2 : \mathbb{R} \longrightarrow \mathbb{R}$ and $c : \mathbb{R}^3 \longrightarrow \mathbb{R}$ as
\begin{equation}
f_1(z) \triangleq - A^{n - 1} D^{-n} \beta \abs{z}^{n - 1} z - A^{n - 1} D^{-n} \gamma \abs{z}^n
\end{equation}
\begin{equation}
f_2(z) \triangleq A^{n - 1} D^{-n} \beta \abs{z}^{n - 1} z - A^{n - 1} D^{-n} \gamma \abs{z}^n
\end{equation}
\begin{equation}
g_1(v) \triangleq \frac{v + \abs{v}}{2}
\end{equation}
\begin{equation}
g_2(v) \triangleq \frac{v - \abs{v}}{2}
\end{equation}
\begin{equation}
h_1(x) \triangleq \alpha \frac{k}{m}x
\end{equation}
\begin{equation}
h_2(z) \triangleq (1 - \alpha) A \frac{k}{m} z
\end{equation}
\begin{equation}
c(x,z,v) \triangleq \frac{b}{m}v
\end{equation}
Then, Eq. \eqref{eq:BWO_rescaled} can be written as 
\begin{equation}\label{eq:BW_main}
\begin{cases}
\dot{x} = v\\
\dot{z} = v + f_1(z) g_1(v) + f_2(z) g_2(v) \\
\dot{v} = - h_1(x) - h_2(z) - c(x, z, v)
\end{cases}
\end{equation}
which is equivalent to Eq. \eqref{eq:main}. It can be verified that the functions $f_1$, $f_2$, $g_1$, $g_2$, $h_1$, $h_2$, $c$ satisfy the conditions necessary for the application of Theorem \ref{thm:main} provided that $A \in \mathbb{R}_{>0}$, $\gamma \in (-\beta, \beta)$, and $b \in \mathbb{R}_{\geq 0}$ (the Bouc-Wen model of class I with $\gamma \neq \beta$). In this case, by Theorem \ref{thm:main}, there exist $x_{\infty}, z_{\infty} \in \mathbb{R}$ such that $\lim_{t \rightarrow +\infty} x(t) = x_{\infty}$, $\lim_{t \rightarrow +\infty} z(t) = z_{\infty}$, $\lim_{t \rightarrow +\infty} v(t) = 0$, and $\alpha x_{\infty} + (1 - \alpha) A z_{\infty} = 0$. This is equivalent to $\lim_{t \rightarrow +\infty} x(t) = x_{\infty}$, $\lim_{t \rightarrow +\infty} z(t) = z_{\infty}$, $\lim_{t \rightarrow +\infty} v(t) = 0$, and $\alpha x_{\infty} + (1 - \alpha) D  z_{\infty} = 0$ in the original coordinates (under the $z_{\infty} \mapsto A D^{-1} z_{\infty}$ transformation).

\section{Conclusions and Future Work}\label{sec:conclusions}

The article described certain analytical properties of an unforced mechanical oscillator with a Duhem-type viscoelastoplastic hysteretic element. These properties include convergence of each solution to an equilibrium point, generalizing a result that was previously presented in Ref. \cite{ikhouane_dynamic_2007}.

Future work may include generalization of the results presented in this article to other (more abstract) Duhem-type models, or mathematical analysis of hysteretic oscillators that utilize the models of hysteresis that cannot be readily described using the Duhem framework (e.g., see \cite{visintin_differential_1994, reid_free_1956, biswas_reduced-order_2014, biswas_two-state_2015, biswas_hysteretic_2016, bhattacharjee_interplay_2017, vaiana_class_2018, vaiana_generalized_2021, vaiana_analytical_2023, vaiana_classification_2023}). It may also be of interest to investigate the stability of individual equilibrium points of the oscillator, estimate the rate of convergence of the solutions, investigate the behavior of the system under the influence of external disturbances. Lastly, it may also be of interest to consider asymptotic behavior of the N-DOF hysteretic oscillators.

\appendices

\section{Notation, Conventions, Foundations}\label{sec:NCF}

Essentially all of the definitions and results that are employed in this article are standard in the fields of set theory, general topology, analysis, ordinary differential equations, and nonlinear systems/control. They can be found in a number of textbooks and monographs on these subjects (e.g., see Ref. \cite{takeuti_introduction_1982}, Refs. \cite{kelley_general_2017, morris_topology_2020}, Refs. \cite{bloch_real_2010, shurman_calculus_2016, ziemer_modern_2017}, Refs. \cite{chicone_ordinary_1999, schaeffer_ordinary_2016}, Refs. \cite{lasalle_extensions_1960, yoshizawa_stability_1966, yoshizawa_stability_1975, sontag_mathematical_1998, sastry_nonlinear_1999, khalil_nonlinear_2002, haddad_nonlinear_2011}, respectively). 

\begin{definition}
$\in$ denotes the set membership relation, $\subseteq$ denotes the subset relation, $\subset$ denotes the proper subset relation, $\cup$ denotes the binary set union operation, $\cap$ denotes the binary set intersection operation, $\setminus$ denotes the binary set difference operation, $\mathcal{P}$ denotes the power set operation, $\emptyset$ denotes the empty set, $(a_1, \ldots, a_n)$ denotes an $n$-tuple, $\{ a_1, \ldots, a_n \}$ denotes an unordered collection of elements.\footnote{It should be noted that some of the syntactic constructions may carry different semantics depending on the context. For example, $(a, b)$ may be used as a pair or as an interval. It is hoped that the context of the discussion will always make the meaning of a given syntactic construction apparent.}
\end{definition}

\begin{definition}
By convention, a topological space cannot be empty. Suppose $X \neq \emptyset$ and $\tau \subseteq \mathcal{P} X$ is a topology on $X$. $\mathsf{cl} A$ denotes the closure of $A \subseteq X$; if $Y \subseteq X$ and $Y \neq \emptyset$, then $\tau | Y$ will denote the subspace topology of $\tau$ on $Y$; the sets $A \subseteq X$ and $B \subseteq X$ are separated if and only if $\mathsf{cl} A \cap B = A \cap \mathsf{cl} B = \emptyset$; a set $C \subseteq X$ is clopen if and only if it is open and closed; $A \subseteq X$ is connected if and only if it is not a union of two nonempty separated sets; $(X, \tau)$ is a connected topological space if and only if $X$ is a connected set.  
\end{definition}

It should be noted that different definitions of a connected set and a connected topological space are employed in some of the cited literature. The following technical lemmas establish a connection between the two commonly used definitions (these results are not used directly, and the proofs were deemed to be sufficiently simple to be omitted):
\begin{lemma}
Suppose $(X, \tau)$ is a topological space. Then, $(X, \tau)$ is connected if and only if the only clopen sets in $(X, \tau)$ are $\emptyset$ and $X$.
\end{lemma}
\begin{lemma}\label{thm:connected_set_alt_def}
Suppose $(X, \tau)$ is a topological space and $Y \subseteq X$. Then, $Y$ is a connected set in $(X, \tau)$ if and only if either $Y = \emptyset$ or $(Y, \tau | Y)$ is a connected topological space.
\end{lemma}

The proof of the following lemma was deemed to be sufficiently simple to be omitted:
\begin{lemma}\label{thm:con_sep}
Suppose $(X, \tau)$ is a topological space. Suppose that $A \subseteq X$ and $B \subseteq X$ are separated, $C \subseteq A \cup B$ is connected. Then, $C \subseteq A$ or $C \subseteq B$. 
\end{lemma}

\begin{definition}
$\mathbb{Z}$ is the set of all integers; $\mathbb{R}$ is the set of all real numbers; an interval of real numbers $I \subseteq \mathbb{R}$ is non-degenerate if it has a non-empty interior; $\mathbb{K}_{>a} \triangleq (a, +\infty) \cap \mathbb{K}$, $\mathbb{K}_{<a} \triangleq (-\infty, a) \cap \mathbb{K}$, $\mathbb{K}_{\geq a} \triangleq [a, +\infty) \cap \mathbb{K}$, and $\mathbb{K}_{\leq a} \triangleq (-\infty, a] \cap \mathbb{K}$ for any $a \in \mathbb{R}$ with $\mathbb{K} \subseteq \mathbb{R}$; $\mathbb{R}^n$ with $n \in \mathbb{Z}_{\geq 1}$ is the set of $n$-tuples of real numbers (augmented with the structure of the Euclidean space); if $X = (x_1, \ldots, x_n) \in \mathbb{R}^n$ with $n \in \mathbb{Z}_{\geq 1}$, then $X_i \triangleq x_i$ for all $i \in \{ 1, \ldots, n \}$; $f : X \longrightarrow Y$ denotes a function with the domain $X$ and the codomain $Y$; given $f : X \longrightarrow Y$, $f(A)$ denotes the image of $A$ under $f$; if $f : X \longrightarrow Y$, then $\mathcal{G}(f) \subseteq X \times Y$ is the graph of $f$ given by $\mathcal{G}(f) \triangleq \{ (x, y) \in X \times Y : y = f(x) \}$; if $f : X \longrightarrow \mathbb{R}^n$ with $n \in \mathbb{Z}_{\geq 1}$, then $f_i : X \longrightarrow \mathbb{R}$ is given by $f_i(x) \triangleq (f(x))_i$ for all $x \in X$ and $i \in \{ 1, \ldots, n \}$; unless stated otherwise, the topology of a subset of $\mathbb{R}^n$ with $n \in \mathbb{Z}_{\geq 1}$ is always the subspace topology of the standard topology on $\mathbb{R}^n$; given $A \subseteq \mathbb{R}$, $\inf A \in \mathbb{R} \cup \{ -\infty, +\infty \}$ denotes the infimum of $A$ and $\sup A \in \mathbb{R} \cup \{ -\infty, +\infty \}$ denotes the supremum of $A$; given a sequence $\{ x_i \in \mathbb{R}^n \}_{i \in \mathbb{Z}_{\geq 1}}$ with $n \in \mathbb{Z}_{\geq 1}$, $\lim_{i \rightarrow +\infty} x_i$ denotes the limit of $x$, provided that it exists; $\langle \cdot, \cdot \rangle : \mathbb{R}^n \times \mathbb{R}^n \longrightarrow \mathbb{R}$ with $n \in \mathbb{Z}_{\geq 1}$ is the canonical inner product on $\mathbb{R}^n$; $\lVert \cdot \rVert_p : \mathbb{R}^n \longrightarrow \mathbb{R}_{\geq 0}$ with $n \in \mathbb{Z}_{\geq 1}$ and $p \in \mathbb{R}_{\geq 1} \cup \{+\infty\}$ is the $p$-norm on $\mathbb{R}^n$; given $n \in \mathbb{Z}_{\geq 1}$ and $p \in \mathbb{R}_{\geq 1} \cup \{+\infty\}$, $d_p : \mathbb{R}^n \times \mathbb{R}^n \longrightarrow \mathbb{R}_{\geq 0}$ given by $d_p(x, y) \triangleq \lVert x - y \rVert_p$ is the metric induced by the $p$-norm, and $d \triangleq d_2$; assuming that $n \in \mathbb{Z}_{\geq 1}$, $p \in \mathbb{R}_{\geq 1} \cup \{+\infty\}$, $a \in \mathbb{R}^n$, and $r \in \mathbb{R}_{>0}$, $\mathbb{B}_p(a, r) \triangleq \{ x \in \mathbb{R}^n : d_p(x, a) < r \}$ is an open $p$-ball in $\mathbb{R}^n$ centered at $a$ with the radius $r$; $f : \mathbb{R}^n \longrightarrow \mathbb{R}^m$ with $n, m \in \mathbb{Z}_{\geq 1}$ is locally Lipschitz if and only if for every $x \in \mathbb{R}^n$ there exists an open set $U \subseteq \mathbb{R}^n$ such that $x \in U$ and there exists $L \in \mathbb{R}_{>0}$ such that $d(f(y),f(z)) \leq L d(y, z)$ for all $y, z \in U$; given a differentiable function $f : X \longrightarrow Y$ such that $X \subseteq \mathbb{R}$ and $Y \subseteq \mathbb{R}^n$ with $n \in \mathbb{Z}_{\geq 1}$, $df/dx$ and $\partial f$ may be used to denote the derivative of the function; given a right-differentiable function $f : X \longrightarrow Y$ such that $X \subseteq \mathbb{R}$ and $Y \subseteq \mathbb{R}^n$ with $n \in \mathbb{Z}_{\geq 1}$, $df/dx^{+}$ and $\partial^{+} f$ may be used to denote the right-derivative of the function; the overdot notation $\dot{x} \triangleq (dx/dt)$ may be used to represents the derivative of a differentiable function $x : \mathbb{R} \longrightarrow \mathbb{R}^n$ with $n \in \mathbb{Z}_{\geq 1}$ with respect to the time variable in the context of mechanics; given a differentiable function $f : X \longrightarrow Y$ such that $X \subseteq \mathbb{R}^n$ and $Y \subseteq \mathbb{R}$ with $n \in \mathbb{Z}_{\geq 1}$, $\partial_i f$ denotes the $i$-th partial derivative of the function for $i \in \{ 1, \ldots, n \}$.
\end{definition}

\begin{lemma}\label{thm:fgx}
Suppose that $f : \mathbb{R} \longrightarrow \mathbb{R}$ is a continuous strictly decreasing function such that $f(0) = 0$. Suppose that $g : \mathbb{R} \longrightarrow \mathbb{R}$ is a continuous strictly increasing function such that $g(a) = 0$ and $g(0) = b$ for some $a \in \mathbb{R}_{\leq 0}$ and $b \in \mathbb{R}_{\geq 0}$. Then, there is a unique $x \in \mathbb{R}$ such that $f(x) = g(x)$. Moreover, $x \in \mathbb{R}_{\leq 0}$.
\end{lemma}
\begin{proof}

Define the continuous function $h : \mathbb{R} \longrightarrow \mathbb{R}$ as 
\[
h(x) \triangleq f(x) - g(x)
\] 
for all $x \in \mathbb{R}$. Note that $h$ is a strictly decreasing continuous function, so it has at most one zero. It remains to show that $h$ has a zero in $\mathbb{R}_{\leq 0}$.

Suppose that $a = 0$. Then, $g(0) = 0$. Since $f(0) = 0$ by assumptions of the lemma, $x = 0 \in \mathbb{R}_{\leq 0}$ is such that $f(x) = g(x)$. Suppose that $b = 0$. Then, $g(0) = 0$. Since $f(0) = 0$ by assumptions of the lemma, $x = 0 \in  \mathbb{R}_{\leq 0}$ is such that $f(x) = g(x)$.

Suppose that $a < 0$ and $b > 0$. Note that $h(a) = f(a) - g(a) = f(a) > 0$. Note also that $h(0) = f(0) - g(0) = -b < 0$. By the Intermediate Value Theorem (e.g., see Theorem 3.5.2 in Ref. \cite{bloch_real_2010}), obtain $x \in (a, 0) \subseteq \mathbb{R}_{\leq 0}$ such that $h(x) = 0$. 
\end{proof}
\begin{lemma}\label{thm:inf_A_U_B}
Suppose that $A, B \subseteq \mathbb{R}$ are bounded from below and $\inf A \leq \inf B$. Then, $\inf (A \cup B) = \inf A$.
\end{lemma}
\begin{proof}
It is apparent that $\inf (A \cup B) \leq \inf A$. Suppose that 
\[
c \triangleq \inf (A \cup B) < \inf A \triangleq a
\]
Obtain $b \in A \cup B$ such that $b \in [c, a)$. Then, $b < \inf A \leq \inf B$, which results in a contradiction. Thus, $\inf (A \cup B) = \inf A$.
\end{proof}

\begin{definition}
Consider the following system of ordinary differential equations
\begin{equation}\label{eq:sys}
\dot{x} = f(x)
\end{equation}
where $f : \mathbb{R}^n \longrightarrow \mathbb{R}^n$ with $n \in \mathbb{Z}_{\geq 1}$ is a locally Lipschitz continuous state function. Equation \eqref{eq:sys} augmented with an initial condition $x(0) = x_0 \in \mathbb{R}^n$ shall be referred to as an Initial Value Problem (IVP) associated with the system given by Eq. \eqref{eq:sys}. A differentiable function $x : I \longrightarrow \mathbb{R}^n$ with $I \subseteq \mathbb{R}$ being a non-degenerate interval such that $0 \in I$ is a solution of the IVP associated with the system given by Eq. \eqref{eq:sys} with the initial condition $x_0 \in \mathbb{R}^n$ if $x(0) = x_0$ and $\dot{x}(t) = f(x(t))$ for all $t \in I$. 
\end{definition}
The following definitions can be found in Ref. \cite{bhat_nontangency-based_2003} and Ref. \cite{haddad_nonlinear_2011}:
\begin{definition}
For the remainder of this definition, suppose that the system given by Eq. \eqref{eq:sys} has a unique solution defined on $\mathbb{R}_{\geq 0}$ for every initial condition. Suppose that $x : \mathbb{R}_{\geq 0} \longrightarrow \mathbb{R}^n$ is a solution of an IVP associated with the system given by Eq. \eqref{eq:sys} with the initial condition $x(0) = z \in \mathbb{R}^n$. Then, $\mathcal{O}_z^{+} \triangleq \{ x(t) : t \in \mathbb{R}_{\geq 0} \}$ is the positive orbit of $z$. A set $U \subseteq \mathbb{R}^n$ is positively invariant with respect to the system given by Eq. \eqref{eq:sys} if and only if for every solution $x : \mathbb{R}_{\geq 0} \longrightarrow \mathbb{R}^n$ of the IVP with $x(0) = z \in U$, $x(t) \in U$ for all $t \in \mathbb{R}_{\geq 0}$. A set $U \subseteq \mathbb{R}^n$ is negatively invariant with respect to the system given by Eq. \eqref{eq:sys} if and only if for every $z \in U$ and $T \in \mathbb{R}_{> 0}$ there exists a solution $x : [0, T] \longrightarrow U$ of the IVP with $x(T) = z$. A set $U \subseteq \mathbb{R}^n$ is invariant with respect to the system given by Eq. \eqref{eq:sys} if and only if it is positively invariant and negatively invariant with respect to the system given by Eq. \eqref{eq:sys}. Suppose again that $x : \mathbb{R}_{\geq 0} \longrightarrow \mathbb{R}^n$ is a solution of an IVP associated with the system given by Eq. \eqref{eq:sys} with the initial condition $x(0) = z \in \mathbb{R}^n$. Then, $p \in \mathbb{R}^n$ is a positive limit point of $z$ if and only if there exists a nondecreasing sequence $\{ t_n \}_{n \in \mathbb{Z}_{\geq 1}}$ of positive real numbers such that $\lim_{n \rightarrow +\infty} t_n = +\infty$ and $\lim_{n \rightarrow +\infty} x(t_n) = p$. Furthermore, $\mathcal{O}_z^{+\infty} \subseteq \mathbb{R}^n$ shall be used to denote the positive limit set of $z$, that is, the set of all positive limit points of $z$. $\lim_{t \rightarrow +\infty} x (t) = A \subseteq \mathbb{R}^n$ if and only if for every $\varepsilon \in \mathbb{R}_{>0}$ there exists $T \in \mathbb{R}_{>0}$ such that $\inf_{p \in A} d(x(t), p) < \varepsilon$ for all $t \in \mathbb{R}_{>T}$.
\end{definition}
The following definition can be found in Ref. \cite{yoshizawa_stability_1975}:
\begin{definition}
The solutions of the system given by Eq. \eqref{eq:sys} are said to be equibounded if and only if for all $\alpha \in \mathbb{R}_{>0}$ there exists $\beta \in \mathbb{R}_{>0}$ such that $\lVert x(t) \rVert_2 < \beta$ for all $t \in [0, T)$ for every solution $x : [0, T) \longrightarrow \mathbb{R}^n$ with $T \in \mathbb{R}_{>0} \cup \{ +\infty \}$ starting from the initial condition $x(0) = x_0 \in \mathbb{R}^n$ such that $\lVert x_0 \rVert_2 \leq \alpha$.
\end{definition}

\section*{Use of Generative Artificial Intelligence (AI)}

During the development of this work, the authors utilized generative AI (specifically, the version of the Google Gemini 3.0 Pro that was available through the Google AI Pro subscription in December 2025 and January 2026) to assist in the theoretical development of the research. The use of the AI was restricted to finding the proofs of a small number of individual steps of the lemmas, propositions, and theorems. The idea and the outline of the proof (i.e., the statements of the lemmas and theorems) and the vast majority of the individual steps in the proofs of the theorems were devised by the authors without the use of the AI. Therefore, the authors believe that, on this occasion, the AI was used predominantly as a tool for productivity enhancement. The original manuscript was written by hand without the assistance of the AI, except for the abstract, which was largely generated by Google Gemini 3.1 Pro that was available through the Google AI Pro subscription in June 2026 (although, even the abstract included minor amendments introduced manually by the authors).

Later, in November 2026, Claude Opus 5.5 (Anthropic) was used to formalize the results of the article in the Lean 4 proof assistant and to review the manuscript. Together with Claude Fable 5.2 (Anthropic), it was also used for the semantic review described in Section \ref{sec:formalization}. The formalization and the review of the manuscript uncovered several errors, which were subsequently corrected; the wording of some of the corrections was suggested by Claude Opus 5.5. The most significant errors were the following: (1) the proofs were found to require $h_1$ and $h_2$ to be strictly increasing locally Lipschitz homeomorphisms, whereas the earlier version of the manuscript assumed only that they were strictly increasing and locally Lipschitz continuous functions; (2) the application of the results to the Bouc-Wen oscillator was found to cover only the Bouc-Wen model of class I with $\gamma \neq \beta$, rather than the whole of class I. Section \ref{sec:formalization}, which describes the formalization, was also added in November 2026 and, unlike, essentially, the rest of the manuscript, was written with the assistance of Claude Opus 5.5. The authors assume full responsibility for the content of the manuscript.

\section*{Acknowledgment}

The authors would like to acknowledge their families, colleagues, and friends. Special thanks go to the members of staff of Auburn University Libraries for their assistance in finding rare and out-of-print research articles and research monographs. The authors would also like to acknowledge the professional online communities, instructional websites, and various online service providers, especially \url{https://www.adobe.com/acrobat/online/pdf-to-word.html}, \url{https://archive.org/}, \url{https://capitalizemytitle.com}, \url{https://www.overleaf.com}, \url{https://pgfplots.net}, \url{https://proofwiki.org/}, \url{https://www.reddit.com}, \url{https://scholar.google.com}, \url{https://stackexchange.com}, \url{https://stringtranslate.com}, \url{https://www.wikipedia.org}. We also note that the results of some of the calculations that are presented in this article were performed with the assistance of the software Wolfram Mathematica \cite{wolfram_research_inc_mathematica_2023}. Other software that was used to produce this article included Adobe Acrobat Reader, Adobe Digital Editions, DiffMerge, Git, GitLab, Google Chrome, Google Gemini, Grammarly, Jupyter Notebook, LibreOffice, macOS Tahoe, Mamba, Microsoft Outlook, Preview, Safari, TeX Live/MacTeX, Texmaker, and Zotero.

\bibliography{template.bib}

\end{document}